\documentclass{llncs}

\usepackage[english]{babel}
\usepackage[T1]{fontenc}
\usepackage[utf8x]{inputenc}

\usepackage{amsfonts}
\usepackage{amssymb}
\usepackage{amsmath}
\usepackage{tikz}

\usepackage{color}
\usepackage{graphicx}
\usepackage{graphics}
\usepackage{xspace}
\usepackage[ruled,vlined,algo2e]{algorithm2e}

\def\cA{\mathcal{A}}
\def\cB{\mathcal{B}}

\def\cE{\mathcal{E}}

\def\cH{\mathcal{H}}

\def\cO{\mathcal{O}}

\def\cX{\mathcal{X}}

\usetikzlibrary{shapes,backgrounds,arrows,automata,snakes,shadows,positioning}
\tikzstyle{every edge}=[-,shorten >=1pt,auto,thin,draw]
\tikzstyle{bluestate}=[draw,diamond]
\tikzstyle{redstate}=[draw,minimum height = .7cm, minimum width=.7cm]
\tikzstyle{whitestate}=[draw,circle]
\tikzstyle{boxbluestate}=[draw,  diamond]
\tikzstyle{boxredstate}=[draw, minimum height=.7cm]
\tikzstyle{boxwhitestate}=[draw,ellipse]

\title{Incremental complexity of a bi-objective  hypergraph transversal problem}
\titlerunning{A bi-objective  hypergraph transversal problem}

\author{Ricardo Andrade\inst{2} \and Etienne Birmel\'e\inst{1} \and Arnaud Mary\inst{2} \and Thomas Picchetti\inst{1} \and Marie-France Sagot\inst{2}}
\authorrunning{R. Andrade et al.}

\institute{MAP5, UMR CNRS 8145, Universit\'e Paris Descartes
\and Universit\'{e} de Lyon, F-69000, Lyon ; Universit\'{e} Lyon 1 ; CNRS, UMR5558, Laboratoire de Biom\'{e}trie et Biologie Evolutive, F-69622 Villeurbanne, France / INRIA Grenoble Rh\^one-Alpes - ERABLE}

\begin{document}

\maketitle

\begin{abstract}
The hypergraph transversal problem has been intensively studied, from both a theoretical and a practical point of view. In particular, its incremental complexity is known to be quasi-polynomial in general and polynomial for bounded hypergraphs. Recent applications in computational biology however require to solve a generalization of this problem, that we call bi-objective transversal problem. The instance is in this case composed of a pair of hypergraphs $(\cA,\cB)$, and the aim is to find minimal sets which hit all the hyperedges of $\cA$ while intersecting a minimal set of hyperedges of $\cB$. In this paper, we formalize this problem, link it to a problem on monotone boolean $\wedge - \vee$ formulae of depth $3$ and study its incremental complexity.
\end{abstract}

\section{Introduction}

  Let
$G(V,\cA)$ with $\cA \subseteq 2^V$ be a hypergraph on a finite set $V$. By abuse of language, we may also from now on refer to the hypergraph more simply by $\cA$ only.
A {\em transversal} of $\cA$ is any set $S\subseteq V$ intersecting all hyperedges of $\cA$.
It is straightforward to see that being a transversal is a monotone property on the subsets of $V$, so that the collection of minimal transversals characterizes all of them. This collection is called the {\em dual} or {\em transversal hypergraph} of $\cA$, and is denoted by $tr(\cA)$.

  The problem of computing the transversal hypergraph of any $\cA$ is equivalent to enumerating maximal independent sets in hypergraphs \cite{BEG00} or to solving the Boolean function dualization problem \cite{EMG08}. Furthermore, it has many applications, for instance in artificial intelligence \cite{EG02}.
This problem thus received much attention in the last decades, both from a theoretical and a practical point of view (see \cite{EMG08} for a review).

  The first method was proposed by Berge~\cite{Berge89}, who considered the hyperedges
iteratively, updating the partial solutions obtained at each step. This algorithm may however have to store a high number of partial solutions, and no final solution will be available until the algorithm stops. More recent work thus focused on methods that build the minimal transversals
iteratively, that is which study the following problem \cite{KBE06}:

\begin{problem}{$\bf{DUAL(\cA, \cX)}$}\label{PBdual}
 Given a hypergraph $\cA$ and a set $\cX$ of minimal transversals of $\cA$, prove that $tr(\cA)=\cX$ or find a new minimal transversal in $tr(\cA)\setminus \cX$.
\end{problem}

 The complexity of this problem remains an open question. However, Fredman and Khachiyan~\cite{FK96} showed that it is quasi-polynomial by proposing two algorithms of respective complexities $N^{\cO(\log^2 N)}$ and $N^{o(\log N)}$, where $N = |\cA| + |\cX|$ is the size of the input. We define the dimension $dim(\cA)$ of a hypergraph $\cA$ as the size of its largest hyperedge, and the degree of a vertex as the number of hyperedges it belongs to. For hypergraphs of bounded dimension, the problem is polynomial \cite{EG95} and parallelizable~\cite{BEG00}. It is also polynomial for hypergraphs of bounded degree~\cite{EGM03}. Moreover, the complexity class of the problem does not change if multiple minimal transversals or partial minimal transversals are required~\cite{BGK00}.

 The performance of the algorithms in practice was also studied in several publications. Khachiyan et al.~\cite{KBE06} introduced an algorithm of the same worst-case complexity than the one of Fredman and Khachiyan but with a better performance in practice. More recently, Toda~\cite{Toda13} and Murakami and Uno~\cite{MU14} compared the existing algorithms and proposed new ones which can deal with large scale hypergraphs.
\medskip

 Determining the transversal of a hypergraph has also several already studied or potential applications in computational biology. It was for example proposed for elaborating knock-out strategies in metabolic networks~\cite{HKS08}, the hyperedges representing metabolic pathways whose activity should be suppressed. One may also consider the vertices as genes and a hyperedge as the set of mutated genes in a tumoral tissue. The transversal hypergraph then lists the collection of minimal mutation sets covering all the tumors. The mutation scenarios would be described by sets of genes, rather than by a single ranking of the genes based on the p-value of a statistical over-representation test.

  However, due to the complexity of cellular mechanisms, in both previous cases it appears there actually are two types of hyperedges, some of them having to be intersected while some others should be avoided.
 Indeed, if one wants to
knock-out a given set of metabolic pathways, one needs to maintain the biomass production of the cell in order to avoid cellular death. H\"adicke and Klamt~\cite{HK11}
 introduced thus the notion of {\em constrained minimal cut sets} corresponding to vertex sets hitting all target pathways while avoiding at least $n$ pathways among a prescribed set. An adaptation of the Berge algorithm was proposed and
was
compared to binary integer programming on real data sets~\cite{JNK13}.

Coming back to the tumoral mutation example, a similar bi-objective problem appears. Mutations may indeed not be related to cancer, and the goal is to discriminate driver mutations from so-called back-seat mutations. Bertrand et al.~\cite{BCS15}
show that this  is equivalent to the {\em minimal set cover} and used a greedy approximation algorithm to solve it. An alternative way to deal with the problem would be to use other mutation data on similar but non tumoral tissues, and to look for mutation collections covering all
tumors while covering as few healthy samples as possible.

We therefore propose to consider a bi-objective generalization of the hypergraph transversal problem, in which two distinct hypergraphs represent, respectively, the sets of nodes to hit
and those to avoid, and to search for the minimal sets of vertices fulfilling both criteria.

\section{The bi-objective transversal problem}\label{PBSection}

\subsection{The problem}

We consider two hypergraphs on the same set of vertices $V$. The first hypergraph $\cA$ will be denoted as the red hypergraph and represents the sets of vertices that have to be intersected.  The second hypergraph $\cB$ will be denoted as the blue hypergraph and represents the sets of vertices which should not be intersected if possible. We will represent such an instance as a tripartite graph, as shown in Figure~\ref{tripartitefig}.

For any $S \subset V$, we define

\[ \cA_S = \{ A \in \cA; S\cap A \neq \emptyset \}  \]

and

\[ \cB_S = \{ B \in \cB; S\cap B \neq \emptyset \}  \]

In particular, for $x\in V$, $\cA_x$ and $\cB_x$ denote the sets of red and blue hyperedges that contain $x$.
\smallskip

The problem then becomes:

\begin{problem}{\bf{Bi-objective hypergraph transversal problem}}\label{PBgeneral}
Given a hypergraph $G=(V,\cH)$ and a partition $\cH=\cA \cup \cB$ of its hyperedges, enumerate the sets $S\subset V$ such that:
\begin{enumerate}
\item
$\cA_S = \cA$ and $S$ is minimal for this property;
\item there exists no $S'$ verifying condition 1 and such that $\cB_{S'}$ is a strict subset of $\cB_S$.
\end{enumerate}
Such sets are called
{\em bi-objective minimal transversals} of the couple $(\cA,\cB)$. The collection of bi-objective minimal transversals of $(\cA,\cB)$ is denoted by $btr(\cA,\cB)$.
\end{problem}

 \begin{figure}
  \begin{center}
  \begin{tikzpicture}
            \node[whitestate] (u) at (0,2) {u};
            \node[whitestate] (v) at (2,2) {v};
            \node[whitestate] (w) at (4,2) {w};
            \node[whitestate] (x) at (6,2) {x};
            \node[redstate] (1) at (1,4) {$A_1$};
            \node[redstate] (2) at (3,4) {$A_2$};
            \node[redstate] (3) at (5,4) {$A_3$};
            \node[bluestate] (4) at (1,0) {$B_1$};
            \node[bluestate] (5) at (3,0) {$B_2$};
            \node[bluestate] (6) at (5,0){$B_3$};
            \path (u) edge (1) ;
            \path (u) edge (2) ;
            \path (u) edge (4) ;
            \path (v) edge (3) ;
            \path (v) edge (2) ;
            \path (v) edge (5) ;
            \path (v) edge (4) ;
            \path (w) edge (3) ;
            \path (w) edge (4) ;
            \path (x) edge (5) ;
            \path (x) edge (6) ;
            \path (x) edge (1) ;
            \path (x) edge (3) ;
  \end{tikzpicture}
  \end{center}
  \caption{Tripartite representation of an instance of the problem. The
  circled vertices are the vertices of the hypergraphs. The squared (resp. diamond) vertices represent the hyperedges of $\cA$ (resp. $\cB$).
    Consider the sets  $S= \{u,v\}$ and $T= \{u,w\}$. Both are minimal transversals for $\cA$. However $S$ is not a solution for the
bi-objective problem as $\cB_T$ is a strict subset of $\cB_S$.}

  \label{tripartitefig}
 \end{figure}
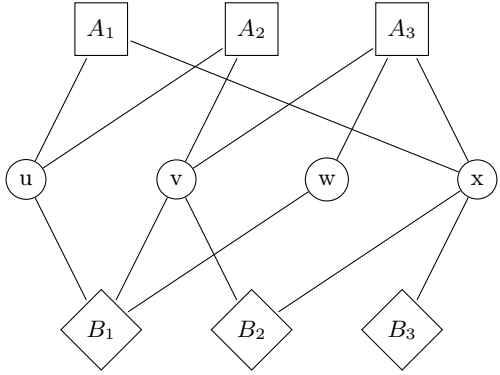

\bigskip

A first approach
for  this problem would be to enumerate all minimal transversals for $\cA$ and then to check for the minimality condition with respect to $\cB$. However, such a procedure may spend an exponential time
on
enumerating minimal transversals of $\cA$ which will be ruled out in the second step. Indeed, consider a hypergraph $\cA$ on a vertex set $V$ such that $\cA$ has an exponential number of minimal transversals. Let $S$ be one of them. Consider the hypergraph $\cB$ having $V\setminus S$ as unique hyperedge. $S$ is then the unique bi-objective minimal transversal of $(\cA,\cB)$.

As the dual hypergraph problem corresponds to the special case $\cB=\emptyset$, we propose to adopt the same strategy of an incremental search of the solutions. We therefore introduce the following problem:

\begin{problem}{$\bf{BIDUAL(\cA, \cB, \cX)}$}\label{PBBIDUAL}
 Given hypergraphs $\cA$ and $\cB$ on the same vertices and a set $\cX$ of bi-objective minimal transversals of $(\cA,\cB)$, prove that $btr(\cA)=\cX$ or find a new minimal transversal in $btr(\cA)\setminus \cX$.
\end{problem}

\subsection{A new enumeration problem}

For $B\subseteq \cB$,
we denote by $S_B$ the set such that
 $S_B=\{x\in V(\cB) \mid \cB_x\subseteq B\}$. In other words, $S_B$ is the set of vertices belonging to no blue hyperedges except for those in $B$.

 For instance, in Figure~\ref{tripartitefig}, $S_{B_1}=\{u,w\}$,  $S_{B_2}=\emptyset$ and $S_{\{B_2,B_3\}}=\{x\}$.

 We define a predicate $f:2^{\cB} \to \{0,1\}$ with

\begin{equation}
  f(B)=\left\{
      \begin{aligned}
        1 & \text{ if }S_B\text{ is a transversal of } \cA\\
        0 & \text{ otherwise}\\
      \end{aligned}
    \right.
\end{equation}

and introduce a new enumeration problem:

\begin{problem}{\bf{Minimal $\cB$-sets enumeration problem }}
\label{PBf}
  Enumerate all minimal $B\subseteq \cB$ such that $f(B)=1$.
\end{problem}

In Figure~\ref{tripartitefig}, $B_1$ is the only solution to the above problem.

The sets to be enumerated in Problems~\ref{PBgeneral} and \ref{PBf} are then linked through the following result.

\begin{lemma}\label{LemmaB}
\begin{enumerate}
\item Let $S$ be a solution to Problem~\ref{PBgeneral}. We then have that $\cB_S$ is a solution to Problem~\ref{PBf}.
\item Let $B$ be a solution to Problem~\ref{PBf}. We then have that every minimal hitting set $S$ of $\cA$ included in $S_B$ is a solution to Problem~\ref{PBgeneral} and $\cB_S=B$.
\end{enumerate}
\end{lemma}

\begin{proof}
The first item is a direct consequence of the definition of a solution to Problem~\ref{PBgeneral}.

For the second item, consider $S'$ a minimal hitting set for $\cA$ with
$\cB_{S'} \subset \cB_S \subset B$. Then $f(\cB_{S'})=1$. By minimality of $B$, $\cB_{S'} = \cB_S = B$.
\qed\end{proof}

Let us then consider the algorithm which enumerates the solutions of Problem~\ref{PBf} and which, each time a new solution $B$ is found,
enumerates all minimal sets $S$ covering $\cA$ and included in $S_B$.


By the first item of Lemma~\ref{LemmaB}, every solution of  Problem~\ref{PBgeneral} is enumerated by this algorithm.

By the second item,  all the sets enumerated by the algorithm are solutions to Problem~\ref{PBgeneral}
and none of them is enumerated twice. Indeed, if that were the case, meaning that two solutions for Problem~\ref{PBf}, say $B$ and $B'$, lead to a same solution $S$ to Problem~\ref{PBgeneral}. The second item then implies that $B=\cB_S=B'$.
%

Moreover, given a set $B\in \cB$,
the enumeration of all minimal sets covering $\cA$ and included in $S_B$ is the transversal hypergraph problem, and can therefore be solved in quasi-polynomial time, and even in polynomial time if $\cA$ is of bounded dimension or if $\cA$ is of bounded degree.
\medskip

A natural question is therefore the complexity class of Problem~\ref{PBf}. The answer to this question is  already partially known. Indeed, consider a set $B$ of hyperedges of $\cB$, and for each hyperedge $b_i$ of $\cB$, the boolean variable $c_i$ indicating if
$b_i\in B$. Then, for any vertex $x$, $x\in S_B$ if and only if
$\bigwedge_{b_i \in \cB_x} c_i = 1$. A hyperedge $a$ of $\cA$ is then hit by $S_B$ if and only if $\bigvee_{x\in a} \bigwedge_{b_i \in \cB_x} c_i = 1$. Finally,

\[ f(B) =  \bigwedge_{a\in \cA} \bigvee_{x\in a} \bigwedge_{b_i \in \cB_x} c_i  \]

Our enumeration problem is therefore the
enumeration of
 all minimal truth assignments satisfying a monotone boolean function of depth $3$.  The situation is illustrated by Figure~\ref{booleanfig}. Conversely, for any monotone formula of type $\bigwedge \bigvee \bigwedge$, one can easily construct an equivalent instance of Problem~\ref{PBf}: simply represent it as a tree of depth 3, merge leaves that are labeled with the same literal, remove the root, and the result is a tripartite graph as in Figure~\ref{tripartitefig}.

 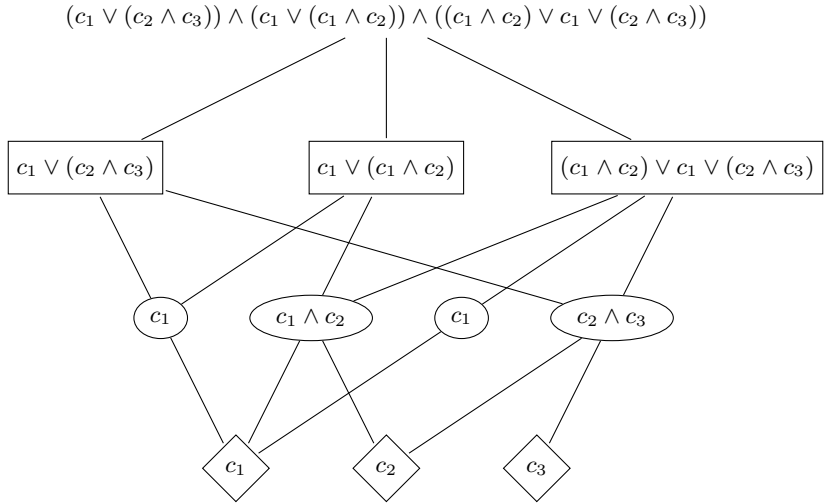
\begin{figure}
  \begin{center}
    \begin{tikzpicture}
            \node[boxwhitestate] (u) at (0,2) {$c_1$};
            \node[boxwhitestate] (v) at (2,2) {$c_1\wedge c_2$};
            \node[boxwhitestate] (w) at (4,2) {$c_1$};
            \node[boxwhitestate] (x) at (6,2) {$c_2 \wedge c_3$};
            \node[boxredstate] (1) at (-1,4) {$c_1 \vee (c_2\wedge c_3) $};
            \node[boxredstate] (2) at (3,4) {$ c_1 \vee (c_1\wedge c_2) $};
            \node[boxredstate] (3) at (7,4) {$(c_1\wedge c_2) \vee c_1 \vee (c_2\wedge c_3)$};
            \node (formula) at (3,6) { $ (c_1 \vee (c_2\wedge c_3)) \wedge (c_1 \vee (c_1\wedge c_2)) \wedge ((c_1\wedge c_2) \vee c_1 \vee (c_2\wedge c_3))$};
            \node[boxbluestate] (4) at (1,0) {$c_1$};
            \node[boxbluestate] (5) at (3,0) {$c_2$};
            \node[boxbluestate] (6) at (5,0){$c_3$};
            \path (u) edge (1) ;
            \path (u) edge (2) ;
            \path (u) edge (4) ;
            \path (v) edge (3) ;
            \path (v) edge (2) ;
            \path (v) edge (5) ;
            \path (v) edge (4) ;
            \path (w) edge (3) ;
            \path (w) edge (4) ;
            \path (x) edge (5) ;
            \path (x) edge (6) ;
            \path (x) edge (1) ;
            \path (x) edge (3) ;
            \path (formula) edge (1) ;
            \path (formula) edge (2) ;
            \path (formula) edge (3) ;
  \end{tikzpicture}
  \end{center}
  \caption{Representation of the instance of Figure~\ref{tripartitefig} in terms of a monotone boolean function.
  The formula of depth $3$ is not true for the truth assignment $c_1=c_2=c_3=0$, indicating that it is not possible to hit all hyperedges of $\cA$ avoiding every hyperedge of $\cB$. However, the truth assignment $c_1=1, c_2=c_3=0$ satisfies the formula, indicating that there exists a solution hitting only the
the first blue (diamond) hyperedge.}

  \label{booleanfig}
 \end{figure}

The associated incremental problem is then the following, introduced in~\cite{GK99}:

\begin{problem}{$\bf{GEN(\mathcal{S})}$} \label{PBGurvich}
 Consider a monotone boolean function of depth $3$ in the shape $\bigwedge \bigvee \bigwedge$.
 Given a set $\mathcal{S}$ of minimal solutions, determine if $S$ is the set of all minimal solutions or find a new one.
\end{problem}

 This problem  was shown in~\cite{GK99} to be coNP-complete. 
{\color{gray}
More precisely, the reduction in~\cite{GK99} shows that in order to
 determine if a DNF formula $d$ is a tautology, it suffices to define a monotone formula $\tilde{d}$ by replacing in $d$ every negative
 litteral $\overline{x_i}$ with a new variable $y_i$, and ask whether the monotone formula $f_0=\tilde{d}\vee[(x_1\vee y_1)\wedge(x_2\vee y_2)\wedge\dots\wedge(x_n\vee y_n)]$ has other prime implicants than the terms in $\tilde{d}$.

 By distributivity $f_0 \equiv f_1=(x_1\vee y_2 \vee \tilde d)\wedge(x_1\vee y_2 \vee \tilde d)\wedge\dots\wedge(x_1\vee y_2 \vee \tilde d)$ which is in the shape $\bigwedge \bigvee \bigwedge$, hence the coNP-completeness of $GEN$.

 Since in the instance of Problem~\ref{PBf} associated to $f_1$ the hypergraphs $\cA$ and $\cB$ have maximum degree $1$ and $3$ (if $d$ is in $3$-DNF), this reduction implies that the corresponding decision problem is coNP-complete even with these degree constraints.

However, the question remains open if the dimension of $\cA$ or $\cB$ is bounded
that is if the
red
(square)
or the blue
(diamond)
vertices in Figure~\ref{booleanfig} are of bounded degree.
}

\medskip

When only one hypergraph is considered, bounding the degree or the dimension reduces the complexity of the transversal enumeration from quasi-polynomial to polynomial.
In the following, we  prove in
Section~\ref{3SATsection}
 that bounding the dimension of $\cB$ or  the degree of the vertices is not enough in the
bi-objective case as
then
$3-SAT$ can be reduced to both $BIDUAL$ and
$GEN$.
However, we show in Section~\ref{boundedAsection} that both problems become polynomial if  the dimension of $\cA$ is bounded.

\section{Bounding the dimension or degree of $\cB$}\label{3SATsection}

\subsection{Results for $\cB$ of bounded dimension}

\begin{theorem}\label{3SATth}
$3-SAT$ can be reduced to $BIDUAL$, even in the case of a bounded value for $dim(\cB)$.
\end{theorem}

\begin{proof}

Let us consider a 3-SAT instance with boolean variables $\{ x_1, \ldots , x_n \}$ and clauses $C_1, \ldots ,C_m$. We can consider, without loss of generality, that there exists no $i$ such that all clauses contain either $x_i$ or $\overline{x_i}$.

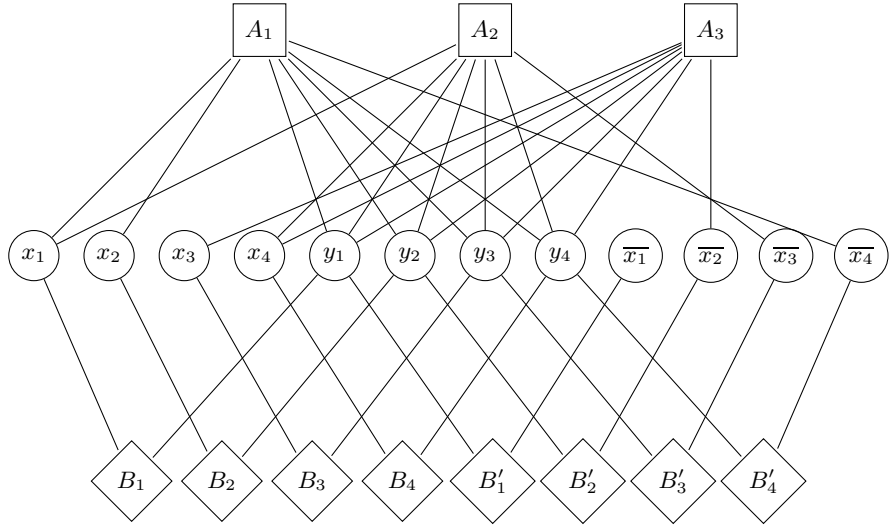
\begin{figure}
  \begin{center}
  \begin{tikzpicture}
            \node[whitestate] (x1) at (0,3) {$x_1$};
            \node[whitestate] (x2) at (1,3) {$x_2$};
            \node[whitestate] (x3) at (2,3) {$x_3$};
            \node[whitestate] (x4) at (3,3) {$x_4$};
            \node[whitestate] (x1bar) at (8,3) {$\overline{x_1}$};
            \node[whitestate] (x2bar) at (9,3) {$\overline{x_2}$};
            \node[whitestate] (x3bar) at (10,3) {$\overline{x_3}$};
            \node[whitestate] (x4bar) at (11,3) {$\overline{x_4}$};
            \node[whitestate] (y1) at (4,3) {$y_1$};
            \node[whitestate] (y2) at (5,3) {$y_2$};
            \node[whitestate] (y3) at (6,3) {$y_3$};
            \node[whitestate] (y4) at (7,3) {$y_4$};
            \node[redstate] (A1) at (3,6) {$A_1$};
            \node[redstate] (A2) at (6,6) {$A_2$};
            \node[redstate] (A3) at (9,6) {$A_3$};
            \node[bluestate] (B1) at (1.3,0) {$B_1$};
            \node[bluestate] (B2) at (2.5,0) {$B_2$};
            \node[bluestate] (B3) at (3.7,0){$B_3$};
            \node[bluestate] (B4) at (4.9,0){$B_4$};
            \node[bluestate] (B1prime) at (6.1,0) {$B'_1$};
            \node[bluestate] (B2prime) at (7.3,0) {$B'_2$};
            \node[bluestate] (B3prime) at (8.5,0){$B'_3$};
            \node[bluestate] (B4prime) at (9.7,0){$B'_4$};
            \path (x1) edge (B1) ;
            \path (y1) edge (B1) ;
            \path (x1bar) edge (B1prime) ;
            \path (y1) edge (B1prime) ;
            \path (x2) edge (B2) ;
            \path (y2) edge (B2) ;
            \path (x2bar) edge (B2prime) ;
            \path (y2) edge (B2prime) ;
            \path (x3) edge (B3) ;
            \path (y3) edge (B3) ;
            \path (x3bar) edge (B3prime) ;
            \path (y3) edge (B3prime) ;
            \path (x4) edge (B4) ;
            \path (y4) edge (B4) ;
            \path (x4bar) edge (B4prime) ;
            \path (y4) edge (B4prime) ;
            \path (x1) edge (A1) ;
            \path (x2) edge (A1) ;
            \path (x4bar) edge (A1) ;
            \path (x1) edge (A2) ;
            \path (x3bar) edge (A2) ;
            \path (x4) edge (A2) ;
            \path (x2bar) edge (A3) ;
            \path (x3) edge (A3) ;
            \path (x4) edge (A3) ;
            \path (y1) edge (A1) ;
            \path (y1) edge (A2) ;
            \path (y1) edge (A3) ;
            \path (y2) edge (A1) ;
            \path (y2) edge (A2) ;
            \path (y2) edge (A3) ;
            \path (y3) edge (A1) ;
            \path (y3) edge (A2) ;
            \path (y3) edge (A3) ;
            \path (y4) edge (A1) ;
            \path (y4) edge (A2) ;
            \path (y4) edge (A3) ;
  \end{tikzpicture}
  \end{center}
  \caption{Instance for the reduction of 3-SAT to $BIDUAL$. The considered clauses are $C_1=x_1 \vee x_2 \vee \overline{x_4}$, $C_2=x_1 \vee \overline{x_3} \vee x_4$ and $C_3= \overline{x_2} \vee x_3 \vee x_4$.}
  \label{3SATfig}
 \end{figure}

 Construct the following hypergraph (see Figure~\ref{3SATfig}):
\begin{enumerate}
\item Consider $3n$ vertices $V = \{ x_1, \overline{x_1} , \ldots, x_n, \overline{x_n}, y_1, \ldots , y_n \} $.
\item For every $1\leq j\leq m$, define a red hyperedge $A_j$ including the $x_i$'s and $\overline{x_i}$'s defining $C_j$ as well as  $\{y_1, \ldots , y_n \}$.
\item For every $1\leq i\leq n$, define a blue hyperedge $B_i = \{x_i, y_i \}$ and a blue hyperedge $ B'_i = \{ \overline{x_i}, y_i \}$.
Observe that $\cB$ is of dimension $2$.
\end{enumerate}

For every $1\leq i\leq n$, consider  $ S_i=\{y_i \}$. It covers all the red hyperedges as well as $B_i$ and $B'_i$. As neither $x_i$ nor $\overline{x_i}$ is contained in all the clauses, it is a minimal solution to the
bi-objective problem.

Consider  the problem $BIDUAL(\cA,\cB,\cX)$ for $\cX = \{  \{y_1 \}, \ldots, \{y_n\} \}$.

Suppose that there exists a minimal bi-objective transversal $S$ in $btr(\cA,\cB) \setminus \cX$. For any $1\leq i\leq n$, $S$ cannot contain both vertices $x_i$ and $\overline{x_i}$. Indeed, it would then cover $B_i$ and $B'_i$, implying $\cB_{S_i} \subset \cB_S$. As at least one clause contains neither $x_i$ nor $\overline{x_i}$, this inclusion would be strict and contradict the minimality of $S$ with respect to the second
bi-objective criterion.

 The former implies that $S$ corresponds to a truth assignment of the boolean variables $x_i$. The fact that it covers $\cA$ is then equivalent to the fact that all clauses are satisfied.
Deciding if there is such  $S$ is therefore NP-Complete, i.e. deciding whether $btr(\cA,\cB) = \cX$ is coNP-complete.
\qed\end{proof}

\begin{remark}
The same reduction is an alternative proof to \cite{GK99} of the coNP-hardness of $GEN$ by considering $\mathcal{S} = \bigcup_{1\leq i\leq n} (B_i,B'_i)$. A same reasoning can then be
applied  to show that finding a new solution would be equivalent to finding a truth assignment satisfying all clauses.
\end{remark}

\subsection{Results for vertices of bounded degree}

The reduction of the
former paragraph
 can be slightly adapted to prove that both $BIDUAL$ and $GEN$ remain hard if the degree of the vertices is bounded rather than the dimension of $\cB$. This is a strong difference with the traditional transversal problem as $DUAL$ becomes polynomial in this case.

\begin{theorem}
$3-SAT$ can be reduced to $BIDUAL$ and $GEN$ in the case of a bounded degree of the vertices of $V$.
\end{theorem}

\begin{proof}
The reduction is similar to the one used to prove Theorem~\ref{3SATth}. Two modifications are done:

\begin{enumerate}
\item The vertices of the hypergraph are split: each $y_i$ is substituted by $m$ vertices $y_i^j, 1\leq j\leq m$, $y_i^j$ belonging to the hyperedges $B_i$, $B'_i$ and $A_j$. Similarly, each $x_i$ (resp. $\overline{x_i}$) is split in as many vertices as the number of red hyperedges it belongs to, and each of them belongs to $B_i$ (resp. $B'_i$) and one of the copies.
\item For each $1\leq i\leq n$, a new blue hyperedge $B''_i$ is created, containing all  $x_i^j$ and $\overline{x_i^j}$.
\end{enumerate}

Figure~\ref{Splitfig} gives an example of this reduction.

\begin{figure}
  \begin{center}
  \begin{tikzpicture}
            \node[whitestate] (x11) at (1,3) {$x_1^1$};
            \node[whitestate] (x12) at (2,3) {$x_1^2$};
            \node             (x2) at (3,3) {$\ldots$};
            \node[whitestate] (x1bar) at (8,3) {$\overline{x_1}$};
            \node             (x2bar) at (9,3) {$\ldots $};
            \node[whitestate] (y11) at (4,3) {$y_1^1$};
            \node[whitestate] (y12) at (5,3) {$y_1^2$};
            \node[whitestate] (y13) at (6,3) {$y_1^3$};
            \node             (y2) at (7,3) {$\ldots$};
            \node[redstate] (A1) at (3,6) {$A_1$};
            \node[redstate] (A2) at (6,6) {$A_2$};
            \node[redstate] (A3) at (9,6) {$A_3$};
            \node[bluestate] (B1) at (-1,0.2) {$B_1$};
            \node[bluestate] (B2) at (0,0) {$B_2$};
            \node[bluestate] (B3) at (1,0.2){$B_3$};
            \node[bluestate] (B4) at (2,0){$B_4$};
            \node[bluestate] (B1prime) at (3,0.2) {$B'_1$};
            \node[bluestate] (B2prime) at (4,0) {$B'_2$};
            \node[bluestate] (B3prime) at (5,0.2){$B'_3$};
            \node[bluestate] (B4prime) at (6,0){$B'_4$};
            \node[bluestate] (B1second) at (7,0.2) {$B''_1$};
            \node[bluestate] (B2second) at (8,0) {$B''_2$};
            \node[bluestate] (B3second) at (9,0.2){$B''_3$};
            \node[bluestate] (B4second) at (10,0){$B''_4$};
            \path (x11) edge (B1) ;
            \path (x12) edge (B1) ;
            \path (x11) edge (B1second) ;
            \path (x12) edge (B1second) ;
            \path (y11) edge (B1) ;
            \path (y12) edge (B1) ;
            \path (y13) edge (B1) ;
            \path (y11) edge (B1prime) ;
            \path (y12) edge (B1prime) ;
            \path (y13) edge (B1prime) ;
            \path (x1bar) edge (B1prime) ;
            \path (x1bar) edge (B1second) ;
            \path (x11) edge (A1) ;
            \path (x12) edge (A2) ;
            \path (y11) edge (A1) ;
            \path (y12) edge (A2) ;
            \path (y13) edge (A3) ;
  \end{tikzpicture}
  \end{center}
  \caption{Reduction of 3-SAT to $BIDUAL$. The considered clauses are $C_1=x_1 \vee x_2 \vee \overline{x_4}$, $C_2=x_1 \vee \overline{x_3} \vee x_4$ and $C_3= \overline{x_2} \vee x_3 \vee x_4$. For readability reasons, only the vertices corresponding to the index $i=1$ are represented in the intermediate layer.}
  \label{Splitfig}
 \end{figure}
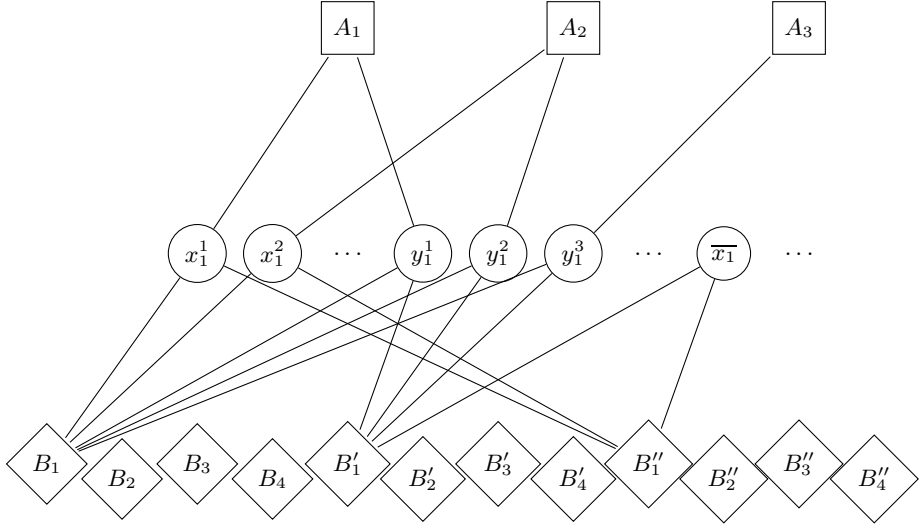

It is easy to see that each vertex of $V$ is of degree exactly $3$. Moreover, for each $i$, $S_i=\{ y_i^1, \ldots, y_i^m \}$ is a solution to the
bi-objective problem and $\cB_{S_i} = \{B_i,B'_i\}$. Consider then the problem $BIDUAL(\cA,\cB,\cX)$ where $\cX$ is the set of all those solutions.

Let $S$ be a new solution of the
bi-objective problem. The set $S$ cannot, for any indices $i,j,k$, contain both $y_i^k$ and a vertex of index $j$. Indeed, $\cB_S$ would then contain $B_i$, $B'_i$ and one set among $B_j$, $B'_j$ and $B''_j$, and therefore also contain $\cB_{S_i}$ as a proper subset. $S$ thus contains no vertex $y$.

The same argument holds if $S$ contains both $x_i^k$ and $\overline{x_i^l}$ for some indices $i,k,l$. Indeed, $\cB_S$ then contains $B_i$, $B'_i$ and  $B''_i$ and consequently contains $\cB_{S_i}$ as a proper subset.

 For every index $i$, $S$ therefore contains either only vertices of the form $x_i^k$ or only vertices of the form $\overline{x_i^k}$. As $S$ covers all the red hyperedges, assigning a value to the variable $x_i$ in the clauses according to this criterion yields a truth assignment satisfying all the clauses.
\smallskip

The reduction from $3-SAT$ to $GEN$ can be done with exactly the same arguments, considering as already known the set $\mathcal{S}$ of solutions of the form $\{B_i, B'_i\}$.

\qed\end{proof}

\section{Bounding the dimension of $\cA$}\label{boundedAsection}

We show in this section that one can reduce Problem \ref{PBgeneral} to the transversal hypergraph problem in a hypergraph of bounded dimension whenever $\cA$ is of bounded dimension.

\begin{theorem}\label{THboundedA}
$GEN$ can be solved in  polynomial time if $\cA$ is of bounded dimension.
As a consequence, Problem~\ref{PBgeneral} can be solved in incremental polynomial time.
\end{theorem}

When the dimension of $\cA$ is bounded, given a set $B \in \cB$, the enumeration
of all minimal sets covering $\cA$ and included in $S_B$ corresponds to the enumeration of minimal transversals of a hypergraph of bounded dimension, and therefore can be done in incremental polynomial time. To show that Problem \ref{PBgeneral} is polynomial, it is sufficient to show that Problem \ref{PBf} is polynomial. Actually we will show that Problem \ref{PBf} can be also reduced to the enumeration of
all minimal transversals of a hypergraph of bounded dimension.

\begin{definition}
For $A\in \cA$, let us denote by $\cH_A$ the hypergraph such that:
          \begin{itemize}
              \item $V(\cH_A)=\bigcup\limits_{x\in A} \cB_{x}$
              \item $\cE(\cH_A)=\{\cB_{x} \mid x\in A\}$
          \end{itemize}
\end{definition}

The following proposition gives a characterisation of
the
subsets $B$ of $\cB$  for which $S_B$ covers a given hyperedge $A\in \cA$.
Given $A\in \cA$, by construction of the hypergraph  $\cH_A$, $S_B$
intersects $\cH_A$ if and only if $B$ contains
a hyperedge of $\cH_A$.

For our purpose, we need to reformulate this simple fact. The formulation given in Lemma \ref{lem:trans}  can be seen as a direct consequence of the following observation.
A subset of vertices $X$ of a hypergraph $\cH$ contains
a hyperedge if and only if $X$ is a transversal of the hypergraph $tr(\cH)$. Indeed, by the duality property between a hypergraph and its transversasl hypergraph (see \cite{Berge89}), the minimal transversals of $tr(\cH)$ are exactly the minimal hyperedges of $\cH$. Thus, a subset of vertices contains
a hyperedge of $\cH$ if and only if it contains a minimal transversal of $tr(\cH)$, i.e. if it is a transversal of $tr(\cH)$.

\begin{lemma}\label{lem:trans}
          Let $B\subseteq \cB$ and $A\in \cA$. We then have that $S_B \cap A\neq \emptyset$ if and only if $B$ is a transversal of $tr(\cH_A)$.
\end{lemma}

\begin{proof}
        ($\Rightarrow$)  Assume that there exists $x \in S_B \cap A$. Then by definition of $S_B$, $\cB_x\subseteq B$. Let $t\in tr(\cH_A)$. Since $\cB_x$ is a hyperedge of $\cH_A$, $t$ must intersect $\cB_x$  and then $t\cap B\neq \emptyset$. We conclude that $B$ is a transversal of $tr(\cH_A)$.

        \medskip

        \noindent ($\Leftarrow$)  Assume now that $B$ is a transversal of $tr(\cH_A)$ and that $S_B \cap A= \emptyset$. By definition of $S_B$, for all $x\in A$, there exists $b_x\in \cB_x$ such that $b_x \notin B$. Let $t\subseteq \cB$ be the set formed by all $b_x$ for all $x\in A$ i.e. $t:=\bigcup\limits_{x \in A} \cB_x \setminus  B$. Since  for all $x\in A$,  $\cB_x \setminus  B \neq \emptyset$, $t$ is a transversal of $\cH_A$ and then contains a minimal transversal $t'$ of $\cH_A$. However by construction of $t$, we have $t'\cap B=\emptyset$, contradicting the fact that $B$ is a transversal of $tr(\cH_A)$.
\qed\end{proof}

Now since we require that $S_B$ covers all hyperedges of $\cA$, $B$ must be a
transversal of $\cH_A$ for every $A\in \cA$.

\begin{proposition}
          $\min\limits_{\subseteq}\{B \subseteq \cB \mid f(B)=1\}=tr(\bigcup\limits_{A\in \cA} tr(\cH_A))$.
\end{proposition}

 \begin{proof}

Let $\cH = \bigcup\limits_{A\in \cA} tr(\cH_A)$.

\begin{align*}
   f(B)=1  &  \Longleftrightarrow  \forall A\in \cA,  S_B \cap A \neq \emptyset \quad \mbox{ by definition} \\
           &  \Longleftrightarrow \forall A\in \cA, B \mbox{ is a transversal of } tr(\cH_A)   \quad \mbox{ by Lemma \ref{lem:trans}} \\
           & \Longleftrightarrow   B  \mbox{ is a transversal of } \cH
\end{align*}

        Thus, the set $\{B \subseteq \cB \mid f(B)=1\}$ is exactly the set of transversals of $\cH$. Therefore,  $\min\limits_{\subseteq}\{B \subseteq \cB \mid f(B)=1\}=tr(\bigcup\limits_{A\in \cA} tr(\cH_A))$.
 \qed\end{proof}

If $dim(\cA)=C$ is bounded, for all $A\in \cA$, $\cH_{A}$ has at most $C$ hyperedges and then each minimal transversal $t$ of $\cH_A$ is of size at most $C$. Then $\bigcup\limits_{A\in \cA} tr(\cH_A))$ is a hypergraph of dimension at most $C$ having at most $|\cA||\cB|^{C}$ hyperedges. We can then construct it in polynomial time and enumerate its minimal transversals in incremental polynomial time using the result for the dualization problem \cite{EG95}.

This implies that if $\cA$ is of bounded dimension,
     $\min\limits_{\subseteq}\{B \subseteq \cB \mid f(B)=1\}$ can be enumerated in incremental polynomial time, thus proving Theorem~\ref{THboundedA}.

\bibliographystyle{acm}
\bibliography{references}

\end{document}